\documentclass[conference]{IEEEtran}
\IEEEoverridecommandlockouts
\usepackage{cite}
\usepackage{amsmath,amssymb,amsfonts}
\usepackage{amsthm}%
\usepackage{algorithm}%
\usepackage{algorithmicx}%
\usepackage{algpseudocode}%
\usepackage{graphicx}
\usepackage{textcomp}
\usepackage{xcolor}
\usepackage{hyperref}

\usepackage{bbm}
\usepackage{bm}
\usepackage{braket}

\newtheorem{theorem}{Theorem}

\newtheorem{definition}{Definition}

\newcommand{\proc}[1]{\textup{\textsf{#1}}}

\def\BibTeX{{\rm B\kern-.05em{\sc i\kern-.025em b}\kern-.08em
    T\kern-.1667em\lower.7ex\hbox{E}\kern-.125emX}}
\begin{document}

\title{Quantum algorithm for\\copula-based risk aggregation\\using orthogonal series density estimation
}

\author{\IEEEauthorblockN{1\textsuperscript{st} Hitomi Mori}
\IEEEauthorblockA{\textit{Graduate School of Engineering Science} \\
\textit{Osaka University}\\
Toyonaka, Osaka, Japan \\
hmori.academic@gmail.com}
\and
\IEEEauthorblockN{2\textsuperscript{nd} Koichi Miyamoto}
\IEEEauthorblockA{\textit{Center for Quantum Information and Quantum Biology} \\
\textit{Osaka University}\\
Toyonaka, Osaka, Japan \\
miyamoto.kouichi.qiqb@osaka-u.ac.jp}
}

\maketitle

\begin{abstract}
Quantum Monte Carlo integration (QMCI) provides a quadratic speed-up over its classical counterpart, and its applications have been investigated in various fields, including finance. This paper considers its application to risk aggregation, one of the most important numerical tasks in financial risk management. Risk aggregation combines several risk variables and quantifies the total amount of risk, taking into account the correlation among them. For this task, there exists a useful tool called copula, with which the joint distribution can be generated from marginal distributions with a flexible correlation structure. Classically, the copula-based method utilizes sampling of risk variables. However, this procedure is not directly applicable to the quantum setting, where sampled values are not stored as classical data, and thus no efficient quantum algorithm is known. In this paper, we introduce a quantum algorithm for copula-based risk aggregation that is compatible with QMCI. In our algorithm, we first estimate each marginal distribution as a series of orthogonal functions, where the coefficients can be calculated with QMCI. Then, by plugging the marginal distributions into the copula and obtaining the joint distribution, we estimate risk measures using QMCI again. With this algorithm, nearly quadratic quantum speed-up can be obtained for sufficiently smooth marginal distributions.
\end{abstract}


\section{Introduction}\label{sec1}
Monte Carlo integration is a widely used method to perform numerical integration using random sampling. It is known that quantum Monte Carlo integration (QMCI) provides a quadratic speed-up over its classical counterpart \cite{Montanaro_2015}.
Applications of QMCI have been considered across various fields such as statistical physics\cite{Montanaro_2015,Cornelissen2023,PhysRevD.107.114511}, nuclear physics\cite{PhysRevD.109.076025}, graph theory\cite{Montanaro_2015}, and machine learning\cite{pmlr-v139-wang21w,wiedemann2023quantum,Wan_Zhang_Li_Zhang_Sun_2023,NEURIPS2023_401aa72e,hikima2024quantum}, and among them, applications in finance have been investigated particularly intensively.
In finance, Monte Carlo integration is used in various tasks. Studies have explored potential applications of QMCI in risk analysis \cite{Woerner_2019,Egger2021,miyamoto2022quantum,wilkens2023quantum} and derivative pricing \cite{Rebentrost_2018,Stamatopoulos_2020,9618807,Stamatopoulos_2022}\footnote{See also a comprehensive review \cite{herman2023quantum} on financial applications of quantum computing including not only QMCI but also optimization and machine learning.}.

Risk aggregation stands as a paramount objective within risk analysis and plays a crucial role in the management of financial institutions \cite{embrechts2002correlation}. Risk aggregation is preceded by risk analysis, wherein substantial risks for a company are identified and quantified. 
For example, insurance companies typically have two major risk categories: underwriting risk and investment risk. Underwriting risk reflects the uncertainty of the total amount of loss that may occur from their insurance contracts, while investment risk arises from the investment of insurance premiums. Furthermore, these categories are broken down into more granular risk categories.
Each risk is then quantified using the loss distribution obtained through Monte Carlo simulation with an appropriate stochastic model. The risk amount can be determined as the quantile of the distribution.
In risk aggregation, these simulated loss distributions are combined to derive the total amount of risk, considering the correlation among different risks\footnote{When risks have perfect correlation, meaning the correlation coefficients are 1, the total risk amount is just the sum of each risk amount. However, in general, the correlation coefficients are within $[-1,1]$, and thus the total risk amount can be smaller than the simple sum, benefiting from the so-called diversification effect.}.
Consequently, a company can obtain the total risk amount representing a certain level (e.g. 1-in-100 year) of potential total loss. By comparing the total risk amount with its capital, the company can assess the sufficiency of its capital. Moreover, the company can calculate other related metrics\footnote{Such as return on risk (the ratio of return to risk amount) and risk contribution (the amount contributing to total risk amount).} for each business unit to develop business plans and determine capital allocation.

To incorporate the correlation structure into the modelling, a statistical tool called copula \cite{Nels06} is widely used in finance. Copula is a function that takes marginal distribution functions of multiple random variables and returns the joint distribution function. This allows for the modeling of multiple risk variables with a flexible correlation structure.
When copula is employed within Monte Carlo simulation, where computations are conducted based on sample values, the correlation is reflected by rearranging the samples.
In this process, the initial step involves sorting the sample values of each risk variable stored in memory. Then, the set of sample values is rearranged according to the order specified by the copula. This rearranged set is regarded as a set of samples from the joint distribution.
However, this approach is not directly applicable to QMCI, where sample values are not stored in memory, and thus no efficient quantum algorithm is known. 

In this paper, we introduce a quantum algorithm for copula-based risk aggregation that aligns with QMCI. In our algorithm, assuming access to the oracles to generate quantum states encoding the possible values of the risk variables, we first estimate the marginal distribution functions of each risk variable using orthogonal series density estimation (OSDE) \cite{Efromovich_2010}.
In OSDE, distribution functions are approximated as a series of orthogonal functions with coefficients represented as expected values of these functions. QMCI is then employed to estimate these expected values. Subsequently, we derive the joint distribution by plugging the estimated marginal distribution functions into the copula.
Finally, we estimate the risk measures from the joint distribution using QMCI.
Through this approach, we successfully obtain the quantum speed-up.
The number of queries to the state generation oracle of each risk variable scales on the required accuracy $\epsilon$ as $\widetilde{O}\left(\left(\frac{1}{\epsilon}\right)^{1+\frac{8}{2r-1}}\right)$, where $r$ is a parameter characterizing the smoothness of the marginal distribution functions.
Thus, if the smoothness $r$ is high, the query complexity approaches $\widetilde{O}(1/\epsilon)$, which represents a quadratic improvement over the sample complexity of order $\widetilde{O}(1/\epsilon^2)$ in the classical algorithm.

The organization of this paper is as follows.
In Section \ref{sec2}, we introduce notations and preliminary knowledge on risk aggregation, OSDE, and QMCI. We present the quantum algorithm for risk aggregation in Section \ref{sec:OurQAlgo} along with an analysis on the error and complexity of this algorithm, followed by the conclusion in Section \ref{sec:Concl}.

\section{Preliminaries}\label{sec2}
In this section, we introduce notations and preliminary knowledge used to develop the quantum algorithm in Section \ref{sec:OurQAlgo}.

\subsection{Notations}\label{sec2-1}
Throughout this paper, we use the notations commonly used in statistics. We use upper case letter (e.g. $X$) to denote random variable and use lower case letter (e.g. $x$) to denote the value of the random variable. For a $\mathbb{R}^d$-valued random variable $\bm{X}=(X_1,\cdots,X_d)$, its probability density function (PDF) $f:\mathbb{R}^d\rightarrow\mathbb{R}$ satisfies
\begin{multline}
    \mathrm{Pr}(\alpha_1\le X_1\le\beta_1,\cdots,\alpha_d\le X_d\le\beta_d)\\
    =\int_{\alpha_1}^{\beta_1}dt_1\cdots\int_{\alpha_d}^{\beta_d}dt_d f(t_1,\cdots,t_d),
\end{multline}
and the cumulative distribution function (CDF) $F:\mathbb{R}^d\rightarrow[0,1]$ is
\begin{align}
    F(x_1,\cdots,x_d)&=\mathrm{Pr}(X_1 \le x_1,\cdots,X_d \le x_d)\\
    &=\int_{-\infty}^{x_1}dt_1\cdots\int_{-\infty}^{x_d}dt_d f(t_1,\cdots,t_d).
\end{align}
$\mathbb{E}_{\bm{X}}[\cdot]$ denotes the expected value; for example,
\begin{align}
    \mathbb{E}_{\bm{X}}[\phi(\bm{X})]=\int_{-\infty}^\infty dt_1\cdots\int_{-\infty}^\infty dt_d \phi(t_1,\cdots,t_d)f(t_1,\cdots,t_d).
\end{align}
Estimators are represented by the hat symbol, such as $\hat{a}$.

As a quantum computing platform, we hereafter consider a system of multiple qubits.
When we express a variable $x\in\mathbb{R}$ with $n$ qubits, for $N=2^n$ and $j\in\{0,1,\cdots,N-1\}$, we represent the computational basis as $\ket{j}$ and the value of $x$ on the $j$th grid point as $x^{(j)}$.
We assume that every quantum register has a sufficient number of qubits and thus consider discretization errors to be negligible.

$\mathbbm{1}_C$ denotes the indicator function, which takes 1 if the condition $C$ is satisfied and $0$ otherwise.

For $x\in\mathbb{R}$ and $\epsilon>0$, we say that $x^\prime\in\mathbb{R}$ is an $\epsilon$-approximation of $x$ if $|x-x^\prime|\le\epsilon$.

\subsection{Risk aggregation}\label{sec2-2}
In financial mathematics, the term ``risk" indicates uncertainties regarding potential profits or losses. Risks are quantified according to risk measures, with two widely used types: Value at Risk (VaR) and Tail Value at Risk (TVaR)\footnote{Also known as Conditional Value at Risk (CVaR).}. VaR is a certain quantile of the loss distribution, and it is most common in financial regulations and disclosures. VaR of a random variable $X$ at a confidence level $\alpha\in[0,1]$ is expressed as
\begin{align}
    \mathrm{VaR}_X(\alpha):=\inf\{x\in\mathbb{R}:F(x) \ge \alpha\}=F^{-1}(\alpha).
\end{align}
TVaR is the conditional expected value of the loss in the tail of the distribution that exceeds the VaR. Because this can capture extreme loss events included in the tail, it is preferred for natural catastrophe risks such as hurricane risk. TVaR of $X$ at a confidence level $\alpha$ is expressed as
\begin{align}
    %\mathrm{TVaR}=\frac{1}{1-\alpha}\int^1_{\alpha}\mathrm{VaR}_t(X)dt.
    \mathrm{TVaR}_X(\alpha)&:=\mathbb{E}_X[X|X \ge \mathrm{VaR}_X(\alpha)] \nonumber \\
    &= \frac{1}{1-\alpha}\int^\infty_{\mathrm{VaR}_X(\alpha)} tf(t)dt.
\end{align}

Although risk measures for each risk variable $X_1,\cdots,X_d$ are of practical interest, those for the sum $S=X_1+\cdots+X_d$ are also crucial.
Risk aggregation estimates this total amount of risk by modeling the joint distribution of $X_1,\cdots,X_d$.
For this purpose, copulas are often used, because they enable the modeling of a flexible correlation structure.
A copula is a CDF for a $[0,1]^d$-valued random variable, where each variable is uniformly distributed on $[0,1]$. By using a copula, we can handle the individual univariate marginal distributions and their dependency separately, thanks to Sklar's theorem, which guarantees the consistency between the copula-based joint distribution and each marginal distribution:\\

\begin{theorem}[Sklar's Theorem \cite{Nels06}] For any joint distribution function $F$ for $\mathbb{R}^d$-valued random variable $\bm{X}=(X_1,\cdots,X_d)$ with marginal distributions $F_1,...,F_d$, there exists a copula $C$ such that
\begin{align}
    F(x_1,...,x_d)=C(F_1(x_1),...,F_d(x_d))
    \label{eq:copula}
\end{align}
for any $(x_1,\cdots,x_d)\in\mathbb{R}^d$.
If $F_1,...,F_d$ are continuous, then $C$ is unique. Conversely, for any marginal distributions $F_1,...,F_d$ and $d$-variate copula $C$, the function $F$ defined as in \eqref{eq:copula} is a CDF with marginal distributions $F_1,...,F_d$.
\end{theorem}

\ \\

Defining the density of the copula as
\begin{align}
    c(u_1,...,u_d)=\frac{\partial}{\partial u_1... \partial u_d}C(u_1,...,u_d),
\end{align}
the joint PDF can be expressed as
\begin{align}
    f(x_1,...,x_d)=c(F_1(x_1),...,F_d(x_d))\prod^d_{i=1}f_i(x_i).
\end{align}

The popular choices of copula include Gaussian copula and T-copula, the details of which can be found in appendix \ref{sec:Copula}.

In practice, we often rely on samples for risk calculation. In many cases, explicit formulas for each marginal CDF $F_i$ may not be available, but we can prepare samples of $X_i$\footnote{For example, if the random time evolution of $X_i$ is modeled by a stochastic differential equation (SDE), which is common for market risk variables, we generally do not have an explicit formula for the CDF of $X_i$ at a future time $T$. However, we can generate sample paths of $X_i$ from the present to $T$ using some discretization method of the SDE, such as the Euler-Maruyama method \cite{Maruyama1955}.}.
Using samples, $F_i$ can be estimated based on ordering \cite{MAINIK2015197}, where we consider the $k$-th smallest sample value as the $100\frac{k}{N}$-percentile point of $X_i$. Specifically, denoting the $N$ sample values of $X_i$ sorted in the ascending order by $x_i^{(1)},\cdots,x_i^{(N)}$, we define the percentile function $\tilde{F}_{i}^{-1}:[0,1]\rightarrow\mathbb{R}$ of $X_i$ by
\begin{equation}
    \tilde{F}_{i}^{-1}(u) :=
    \begin{cases}
        x_i^{(1)} & \mathrm{for} \ 0 \le u \le \frac{1}{N}, \\
        x_i^{(2)} & \mathrm{for} \ \frac{1}{N} < u \le \frac{2}{N}, \\
        & \vdots \\
        x_i^{(N)} & \mathrm{for} \ \frac{N-1}{N} < u \le 1
    \end{cases}.
    \label{eq:PercentileFunc}
\end{equation}
When we have $N$ samples for each $X_i$ from $F_i$ and $(U_1,\cdots,U_d)\in[0,1]^d$ from $C$\footnote{The procedure for sampling from a copula can be found in \cite{embrechts2002correlation}.}, we can generate samples from the joint distribution in \eqref{eq:copula} as described in Algorithm \ref{alg:OrdCopula}.
Given the sample values $\bm{x}^{(1)},\cdots,\bm{x}^{(N)}$, we estimate $\mathrm{VaR}_S(\alpha)$ as $\tilde{F}_{S}^{-1}(\alpha)$, which is defined similarly to \eqref{eq:PercentileFunc}, and $\mathrm{TVaR}_S(\alpha)$ as
\begin{equation}
    \frac{\sum_{l=1}^N s^{(l)} \mathbbm{1}_{s^{(l)} \ge \mathrm{VaR}_S(\alpha)}}{\sum_{l=1}^N \mathbbm{1}_{s^{(l)} \ge \mathrm{VaR}_S(\alpha)}},
\end{equation}
where $s^{(l)}:=\sum_{i=1}^d \tilde{x}_i^{(l)}$.
It is shown that samples $\bm{x}^{(l)}$ generated by this method asymptotically obey $F$ in \eqref{eq:copula} with the CDF converging at a rate $O(N^{-1/2})$, and so do the VaR and TVaR estimations from the samples\footnote{See \cite{MAINIK2015197} for mathematically rigorous statements.}.

\begin{algorithm}[b]
\caption{Sampling from the copula-based joint distribution}\label{alg:OrdCopula}
\begin{algorithmic}[1]
\Require
\Statex Samples from each marginal $F_i$ and copula $C$.

\For{$i=1,\cdots,d$}
\State Get $N$ sample values of $X_i \sim F_i$ 
\State Sort samples in the ascending order $x_i^{(1)},\cdots,x_i^{(N)}$.
\EndFor
\For{$l=1,\cdots,N$}
\State Sample $(U_1^{(l)},...,U_d^{(l)}) \sim C$.
\State Let $\bm{x}^{(l)}:=(\tilde{x}^{(l)}_1,\cdots,\tilde{x}^{(l)}_d)$, where $\tilde{x}_i^{(l)}:=\tilde{F}_{i}^{-1}(u_i^{(l)})$.
\EndFor
\State Output $\bm{x}^{(1)},\cdots,\bm{x}^{(N)}$.
\end{algorithmic}
\end{algorithm}

Note that this procedure based on sampling and rearranging is not directly applicable to the quantum algorithm, because QMCI does not use sample values of $X_i$ but rather a superposition of quantum states that encode the possible values of $X_i$, as discussed in Section \ref{sec:QAE}.

\subsection{Orthogonal series density estimation}\label{sec2-3}
Orthogonal series density estimation \cite{Efromovich_2010} is a technique used to estimate the PDF $f$ of a random variable $X\in\mathbb{R}$ using orthogonal series.
As a specific choice, we employ a series of Hermite functions as outlined in \cite{GREBLICKI1984174}:
\begin{equation}
    h_k(x)=\frac{1}{\sqrt{2^k k! \pi^{1/2}}}H_k(x)e^{-x^2/2}, 
\end{equation}
where
\begin{equation}
    H_k=(-1)^ke^{x^2}\frac{d^k}{dx^k}e^{-x^2}
\end{equation}
is the $k$-th Hermite polynomial for $k=0,1,\cdots$, assuming that the range of $X$ (equivalently, the domain of $f$) is $(-\infty,\infty)$\footnote{Other choices of orthogonal series are also possible depending on the domain of $f$, and the change in this choice will not affect the following discussion in any essential way.}. The Hermite function $h_k$ satisfies the following orthogonal relationship:
\begin{align}
    &\int_{-\infty}^\infty h_k(x) h_l(x)dx=\delta_{kl},
\end{align}
where $k,l\in\mathbb{N}$ and $\delta_{kl}$ is the Kronecker's delta.
We approximate $f$ by the Hermite series ${f}^{(K)}$ defined as
\begin{align}
    &{f}^{(K)}(x):=\sum^K_{k=0} a_k h_k(x)
\end{align}
for $K\in\mathbb{N}$ with the coefficients given by
\begin{equation}
    a_k=\int_{-\infty}^\infty f(x) h_k(x)dx.
\end{equation}
Because $f$ is the PDF of the random variable $X$, $a_k$ can be expressed as the expected value
\begin{equation}
    a_k=\mathbb{E}_X[h_k(X)].
\end{equation}
The error bound for the truncated Hermite series is given by the following theorem, which is extracted from the discussion leading to Theorem 5 in \cite{GREBLICKI1984174}.\\

\begin{theorem}[Approximation by Hermite series]

Let $f:\mathbb{R}\rightarrow\mathbb{R}$ be in $L^\infty \cap L^2$
and of bounded variation on any finite interval. Besides, for some $r\in\mathbb{N}$, $\left(x-\frac{d}{dx}\right)^r f \in L^2$ holds.
Then, for any $L>0$, there exists $\gamma_{f,L}\in\mathbb{R}$ such that
\begin{equation}
    \sup_{x\in[-L,L]} \left|{f}^{(K)}(x)-f(x)\right|\le \gamma_{f,L} K^{-\frac{r}{2}+\frac{1}{4}}
    \label{eq:HSErr}
\end{equation}
holds for any $K\in\mathbb{N}$.
\label{th:HSErr}
\end{theorem}

\ \\

The Hermite function also has the following property: for any $k=0,1,\cdots$ and any $x\in\mathbb{R}$,
\begin{equation}
    |h_k(x)|\le \pi^{-1/4}
    \label{eq:HFBnd}
\end{equation}
holds \cite{szeg1939orthogonal}.

\subsection{Quantum Monte Carlo integration}\label{sec:QAE}
We first introduce quantum amplitude estimation (QAE) \cite{Brassard_2002}, which serves as the central subroutine in QMCI and is the source of the quadratic speed-up. QAE is a technique to estimate the squared amplitude of a specific basis state in a superposition, and the error bound of QAE is given by the following theorem.\\

\begin{theorem}[Quantum amplitude estimation \cite{Brassard_2002}]
There is a quantum algorithm called amplitude estimation which takes as input one copy of a quantum state $\ket{\psi}$, a unitary transformation $U=2\ket{\psi}\!\bra{\psi}-I$, a unitary transformation $V=I-2P$ for some projector $P$, and an integer $t$. The algorithm outputs $\tilde{a}$, an estimate of $a=\bra{\psi}P\ket{\psi}$, such that
\begin{align}
    \left|\tilde{a}-a\right|\le2\pi\frac{\sqrt{a(1-a)}}{t}+\frac{\pi^2}{t^2}
\end{align}
with probability %of
at least $8/\pi^2$, using $U$ and $V$ $t$ times each.
\label{th:QAE}
\end{theorem}

\ \\

Using QAE, we can construct QMCI, the quantum algorithm to estimate $\mathbb{E}_{\bm{Y}}[\phi(\bm{Y})]$ for a random variable $\bm{Y}\in\mathbb{R}^d$ and a function $\phi:\mathbb{R}^d\rightarrow[0,1]$.
Because we assume that we have sufficient qubits to represent real numbers with the negligible discretization error, we can express the expected value in terms of summation:

\begin{equation}
    \mathbb{E}_{\bm{Y}}[\phi(\bm{Y})]=\frac{1}{\mathcal{N}_{\bm{Y}}^2}\sum_{\bm{j}\in\mathcal{G}_{\bm{Y}}} \phi(y_1^{(j_1)},\cdots,y_d^{(j_d)})f(y_1^{(j_1)},\cdots,y_d^{(j_d)}),
\end{equation}
where $\mathcal{G}_{\bm{Y}}=\mathcal{G}_{Y_1}\times\cdots\times\mathcal{G}_{Y_d}$, $\mathcal{G}_{Y_i}=\{0,\cdots,N_i-1\}$ (i.e. $\sum_{\bm{j}\in\mathcal{G}_{\bm{Y}}}=\sum_{j_1=0}^{N_1-1}\cdots\sum_{j_d=0}^{N_d-1}$), and $y_i^{(j_i)}$ is the $j_i$-th discretization point in $y_i$ axis.
Here we define two types of oracles required to run QMCI.\\
%The first one is the oracle to generate a quantum state encoding the PDF in the amplitude.\\

\begin{definition}[State preparation]
For a random variable $\bm{Y}=(Y_1,\cdots,Y_d)\in\mathbb{R}^d$ with the PDF $f$, we have access to the unitary operation $\mathcal{A}_{\bm{Y}}$ that acts as
\begin{multline}
    \mathcal{A}_{\bm{Y}}\ket{0}^{\otimes n_1}\cdots\ket{0}^{\otimes n_d}\\
    =\frac{1}{\mathcal{N}_{\bm{Y}}}\sum_{\bm{j}\in\mathcal{G}_{\bm{Y}}}\sqrt{f(\bm{y}^{(\bm{j})})}\ket{j_1}\cdots\ket{j_d},
\end{multline}
where $\bm{y}^{(\bm{j})}=(y_1^{(j_1)},\cdots,y_d^{(j_d)})$ and
$\mathcal{N}_{\bm{Y}}:=\sqrt{\sum_{\bm{j}\in\mathcal{G}_{\bm{Y}}} f(\bm{y}^{(\bm{j})})}$.
\label{ass:StatePrep}
\end{definition}
\ \\

If we know the exact form of the PDF, then the state preparation can be achieved using existing methods \cite{grover2002creating,Sanders2019,rattew2022preparing,holmes2020efficient,mcardle2022quantum,Moosa_2023}. 
Even if we do not know the exact form, given a quantum circuit to simulate ${\bm{Y}}$, we can perform the state preparation.
For example, if ${\bm{Y}}$ represents a value of a stochastic process at a specific time and the process is modeled by an SDE, the state can be prepared as described in \cite{Carrera_Vazquez_2021,Kaneko2022}.\\

\begin{definition}[Controlled rotation]
    For a function $\phi:\mathbb{R}^d\to[0,1]$, we have access to the unitary operation $W_\phi$ that acts as
\begin{multline}
    W_\phi\ket{j_1}\cdots\ket{j_d}\ket{0}\\
    =\ket{j_1}\cdots\ket{j_d}\left(\sqrt{1-\phi(\bm{y}^{(\bm{j})})}\ket{0}+\sqrt{\phi(\bm{y}^{(\bm{j})})}\ket{1}\right).
\end{multline}
for any $\bm{y}^{(\bm{j})}\in\mathbb{R}^d$.
\label{ass:CRot}
\end{definition}

\ \\

Methods to implement this for some specific cases can be found in \cite{Woerner_2019,Stamatopoulos_2020,Carrera_Vazquez_2021}.
If $\phi$ can be expressed as an elementary function, we can compute $\theta=\arcsin\left(\sqrt{\phi(\bm{y})}\right)$ onto an ancillary register using arithmetic circuits \cite{haner2018optimizing,MunozCores2022} and then apply a rotation gate with the angle controlled 
$\ket{\theta}\ket{0}\rightarrow\ket{\theta}(\cos\theta\ket{0}+\sin\theta\ket{1})$.

Applying $W_\phi(\mathcal{A}_{\bm{Y}}\otimes I)$ to the initial state $\ket{0}^{\otimes s}\ket{0}$ with $s=n_1+\cdots+n_d$, we can prepare the state
\begin{multline}
    \ket{\psi}:=\frac{1}{\mathcal{N}_{\bm{Y}}}\sum_{\bm{j}\in\mathcal{G}_{\bm{Y}}}\sqrt{f(\bm{y}^{(\bm{j})})}\ket{j_1}\cdots\ket{j_d}\\
    \left(\sqrt{1-\phi(\bm{y}^{(\bm{j})})}\ket{0}+\sqrt{\phi(\bm{y}^{(\bm{j})})}\ket{1}\right).
\end{multline}
Note that for $P=I^{\otimes s} \otimes \ket{1}\!\bra{1}$,
\begin{equation}
    \braket{\psi | P | \psi}
    =\frac{1}{\mathcal{N}_{\bm{Y}}^2}\sum_{\bm{j}\in\mathcal{G}_{\bm{Y}}} \phi(\bm{y}^{(\bm{j})})f(\bm{y}^{(\bm{j})})
    ={\mathbb{E}}_{\bm{Y}}[\phi(\bm{Y})].
\end{equation}
Thus, we can estimate ${\mathbb{E}}_{\bm{Y}}[\phi(\bm{Y})]$ using QAE with
\begin{align}
    U&=2\ket{\psi}\!\bra{\psi}-I\\
    &=W_\phi(\mathcal{A}_{\bm{Y}}\otimes I)\left(2\ket{0}\!\bra{0}^{\otimes s}\otimes\ket{0}\!\bra{0}-I\right)(W_\phi(\mathcal{A}_{\bm{Y}}\otimes I))^\dagger.
\end{align}
The complexity of this method, QMCI, is summarized as follows.\\

\begin{theorem}[Quantum Monte Carlo integration; Theorem 2.3 in \cite{Montanaro_2015}, modified]
    Let $\delta,\epsilon\in(0,1)$.
    Assuming we have access to the state preparation oracle $\mathcal{A}_{\bm{Y}}$ for a random variable $\bm{Y}\in\mathbb{R}^d$ and the controlled rotation oracle $W_{\phi}$ for a function $\phi:\mathbb{R}^d\rightarrow[0,1]$, there exists a quantum algorithm called $\proc{QMCI}(\mathcal{A}_{\bm{Y}},\phi,\epsilon,\delta)$ that, with probability at least $1-\delta$, outputs an $\epsilon$-approximation of ${\mathbb{E}}_{\bm{Y}}[\phi(\bm{Y})]$, querying $\mathcal{A}_{\bm{Y}}$ and $W_\phi$ $O\left(\frac{1}{\epsilon}\log\left(\frac{1}{\delta}\right)\right)$ times each.
    \label{th:QMCI}
\end{theorem}

\ \\

Note that the success probability is enhanced from $8/\pi^2$ in Theorem \ref{th:QAE} to $1-\delta$ by taking the median of outputs of multiple runs of QAE (Lemma 2.1 in \cite{Montanaro_2015}), which is based on Lemma 6.1 in \cite{jerrum1986random}.

\section{Quantum algorithm for risk aggregation}\label{sec:OurQAlgo}

\subsection{Overview}

In this section, we show our main contribution: the algorithm for risk aggregation.
First, we note that risk aggregation for a set of risk variables $\bm{X}=(X_1,\cdots,X_d)$ can be formulated as the calculation of the expected value in the following form:
\begin{align}
    &\mathbb{E}_{\bm{X}}\left[g\left(\bm{X}\right)\right]\nonumber\\
    & = 
    \frac{1}{\mathcal{N}_{\bm{X}}^2}\sum_{\bm{x}\in\mathcal{G}_{\bm{X}}} g(\bm{x}^{(\bm{j})})f(\bm{x}^{(\bm{j})}) \nonumber \\
    &=\frac{1}{\mathcal{N}_{\bm{X}}^2}\sum_{\bm{x}\in\mathcal{G}_{\bm{X}}} g(\bm{x}^{(\bm{j})})c(F_1(x_1^{(j_1)}),...,F_d(x_d^{(j_d)}))\prod^d_{i=1}f_i(x_i^{(j_i)}) \nonumber \\
    &=\mathbb{E}^\mathrm{ind}_{\bm{X}}[g(\bm{X})c(F_1(X_1),...,F_d(X_d))].
    \label{eq:TgtExp}
\end{align}
Here, $g:\mathbb{R}^d\rightarrow[0,1]$ is a function chosen according to the type of the risk measure, as explained later.
$\mathbb{E}^\mathrm{ind}_{\bm{X}}[\cdot]$ is the expected value under the assumption that $X_1,\cdots,X_d$ are independent of each other, i.e. the PDF is expressed as
\begin{equation}
    f^\mathrm{ind}(x_1,\cdots,x_d)=\prod^d_{i=1}f_i(x_i).
\end{equation}
When we have the state preparation oracles $\mathcal{A}_{X_1},\cdots,\mathcal{A}_{X_d}$ for $X_1,\cdots,X_d$, we can generate the following quantum state
\begin{equation}
    \mathcal{A}_{\bm{X}}^{\rm ind}\ket{0}^{\otimes d}=\frac{1}{\mathcal{N}_{\bm{X}}}\sum_{\bm{j}\in\mathcal{G}_{\bm{X}}}\sqrt{f^\mathrm{ind}(\bm{x}^{(\bm{j})})}\ket{j_1}\cdots\ket{j_d},
\end{equation}
where $\mathcal{A}_{\bm{X}}^{\rm ind}:=\mathcal{A}_{X_1}\otimes\cdots\otimes\mathcal{A}_{X_d}$ and $\mathcal{N}_{\bm{X}}=\prod_{i=1}^d \mathcal{N}_{X_i}$ with $\mathcal{N}_{X_i}:=\sqrt{\sum_{j_i\in\mathcal{G}_{X_i}} f_i(x_i^{(j_i)})}$.
This motivates us to estimate \eqref{eq:TgtExp} by QMCI.
Since we are given explicit formulas for $c$ and $g$ in many cases, the remaining requirement for running QMCI is the formulation of marginal CDFs $F_1,\cdots,F_d$.
To do this, we use OSDE, which gives us explicit approximation formulas of the CDFs as orthogonal series.
In summary, our algorithm is composed of two steps: estimating the marginal CDFs by OSDE and estimating the risk measures.

\subsection{Estimating the marginal distributions}\label{subsec3-1}

Now, we explain the process to estimate the marginal CDFs.
Here, we only consider sufficiently smooth marginal distributions that can be well approximated by Hermite series (see the assumptions in Theorems \ref{th:HSErr}), which is usually the case in risk analysis.
We assume the availability of the state preparation oracle $A_{X_i}$ with $X_i$ and the controlled rotation with $\phi=\bar{h}_k$ for $k=0,1,\cdots$, where
\begin{equation}
    \bar{h}_k(x)=\frac{1}{2}+\frac{\pi^{1/4}}{2}h_k(x),
\end{equation}
which satisfies $0\le\bar{h}_k(x)\le1$.
Then, we run QMCI to estimate ${\mathbb{E}}_{X_i}[\bar{h}_k(X_i)]$
for $k=0,\cdots,K_i$ and get the estimations $\hat{a}^{X_i}_k$ for
\begin{equation}
    a^{X_i}_k:=\mathbb{E}_{X_i}[h_k(X_i)]=\frac{1}{\pi^{1/4}}\left(2{\mathbb{E}}_{X_i}[\bar{h}_k(X_i)]-1\right),
\end{equation}
the exact coefficient in the Hermite expansion for $f$.
Then, given the estimated coefficients, we have a Hermite series approximating $f_i$ as
\begin{equation}
\hat{f}_i(x)=\sum^{K_i}_{k=0}\hat{a}^{X_i}_k h_k(x).
\label{eq:hatf}
\end{equation}
Furthermore, by integrating this, we get an approximation $\hat{F}_i$ for $F_i$.
In fact, since the accuracy of the Hermite series approximation is guaranteed not on the entire real axis but in a finite interval, we set $\hat{F}_i$ to 0 or 1 outside the interval.
Namely, we define
\begin{align}
\hat{F}_i(x):= 
\begin{cases}
0 & ; \ x < -L_i \\ 
\sum^{K_i}_{k=0}\hat{a}^{X_i}_k \mathcal{H}_{k,L_i}(x) & ; -L_i \le x < L_i \\
1 & ; \ x \ge L_i
\end{cases}
\label{eq:hatF}
\end{align}
with sufficiently large $L_i>0$, where
\begin{equation}
    \mathcal{H}_{k,L_i}(x):=\int_{-L_i}^x h_k(t)dt.
\end{equation}
Now, we just assume that we can evaluate $\mathcal{H}_{k,L_i}$ with a negligible error, postponing the discussion on this point to Appendix \ref{sec:HFIntEval}.

The above method is summarized as Algorithm \ref{alg:marg}.

\begin{algorithm}[b]
\caption{$\proc{EstMarg}(X_i,K_i,\epsilon,\delta)$: estimation of $F_i$}\label{alg:marg}
\begin{algorithmic}[1]
\Require
\Statex Access to the state preparation oracle $\mathcal{A}_{X_i}$.
\Statex Access to the rotation oracles $W_{\bar{h}_k}$ for $k=0,\cdots,K_i$.
\Statex $L_i>0$ and $\epsilon,\delta\in(0,1)$.

\For{$k=0,\dots,K_i$}
\State Run $\proc{QMCI}\left(\mathcal{A}_{X_i},\bar{h}_k,\frac{\sqrt{\pi}\epsilon}{16L_i(K_i+1)},\frac{\delta}{K_i+1}\right)$ and let the output be the estimation $\hat{\mathbb{E}}_{X_i}[\bar{h}_k(X_i)]$.
\State Set $\hat{a}^{X_i}_k:=\frac{1}{\pi^{1/4}}\left(2\hat{\mathbb{E}}_{X_i}[\bar{h}_k(X_i)]-1\right)$.
\EndFor

%\State Define $\hat{f}_{X_i}$ as Eq. \eqref{eq:hatf}.

\State Define $\hat{F}_i$ as \eqref{eq:hatF} and output it.

\end{algorithmic}
\end{algorithm}

We give the following theorem on the accuracy and complexity of this method.\\

\begin{theorem}
Let $\delta,\epsilon\in(0,1)$.
Let $X_i$ be a real-valued random variable following the distribution $f_i$. Assume the following:
\begin{enumerate}
\renewcommand{\labelenumi}{(\roman{enumi})}
\item $f_i$ has the properties of Theorem \ref{th:HSErr} for $r=r_i$.
\item Access to the state preparation oracle $\mathcal{A}_{X_i}$.
\item Access to the controlled rotation oracle $W_{\phi}$ with $\phi=\tilde{h}_k$ for $k=0,\cdots,K_i$.
\item For some $L_i>0$, $F_{i}(-L_i)\le \frac{\epsilon}{2}$ and $F_{i}(L_i)\ge 1-\epsilon$ hold.

\end{enumerate}
Then, with probability at least $1-\delta$, we get $\hat{F}_i$ such that 
\begin{equation}
\left|\hat{F}_i(x)-F_i(x)\right|\le\epsilon
\label{eq:FhatErr}
\end{equation}
holds for any $x\in\mathbb{R}$ by $\proc{EstMarg}\left(X_i,K_i,\epsilon,\delta\right)$, which queries $\mathcal{A}_{X_i}$ and $W_{\bar{h}_0},\cdots,W_{\bar{h}_{K_i}}$
\begin{align}
    O\left(\frac{L_iK_i^2}{\epsilon}\log \left(\frac{K_i}{\delta}\right)\right)\label{eq:CompMarg}
\end{align}
times. Here we set
\begin{equation}
    K_i=\left\lceil\left(\frac{8L_i\gamma_i}{\epsilon}\right)^{4/(2r_i-1)}\right\rceil,
    \label{eq:Ki}
\end{equation}
where $\gamma_i$ is a real number such that
\begin{equation}
    \sup_{x\in[-L_i,L_i]} \left|{f}^{(K)}_i(x)-f_i(x)\right|\le \gamma_i K^{-\frac{r_i}{2}+\frac{1}{4}}
    \label{eq:HSErrfXi}
\end{equation}
holds for any $K\in\mathbb{N}$~\footnote{Because of (i), such $\gamma_i$ exists, as implied by Theorem \ref{th:HSErr}.}.
\label{th:marg}
\end{theorem}

\begin{proof}
    Because of the definition of $\hat{F}_i$ in \eqref{eq:hatF} and (iv), it is immediately seen that \eqref{eq:FhatErr} holds for any $x\in(-\infty,-L_i]\cup[L_i,\infty)$.
    Thus, we hereafter focus on the case that $x\in(-L_i,L_i)$.

    We start by evaluating $\left|\hat{f}_i(x)-f_i(x)\right|$.
    Decomposing it as
    \begin{equation}
        \left|\hat{f}_i(x)-f_i(x)\right|\le \left|\hat{f}_i(x)-{f}^{(K_i)}_i(x)\right| + \left|{f}^{(K_i)}_i(x)-f_i(x)\right|,
        \label{eq:TriIneqf}
    \end{equation}
    we bound the first and second terms separately.
    On the second term, for $x\in(-L_i,L_i)$, we have
    \begin{equation}
        \left|{f}^{(K_i)}_i(x)-f_i(x)\right|\le \gamma_i K_i^{-\frac{r_i}{2}+\frac{1}{4}} \le \frac{\epsilon}{8L_i},
        \label{eq:Diff_tilf_f}
    \end{equation}
    where we use \eqref{eq:Ki} and \eqref{eq:HSErrfXi}.
    To bound the first one, we temporarily assume that $\proc{QMCI}\left(\mathcal{A}_{X_i},\bar{h}_k,\frac{\sqrt{\pi}\epsilon}{16L_i(K_i+1)},\frac{\delta}{K_i+1}\right)$ for every $k=0,\cdots,K_i$ outputs the estimation $\hat{\mathbb{E}}_{X_i}[\bar{h}_k(X_i)]$ such that
    \begin{equation}
        \left|\hat{\mathbb{E}}_{X_i}[\bar{h}_k(X_i)]-{\mathbb{E}}_{X_i}[\bar{h}_k(X_i)]\right|\le \frac{\sqrt{\pi}\epsilon}{16L_i(K_i+1)}.
        \label{eq:alphaErr}
    \end{equation}
    This implies that
    \begin{align}
        \left| \hat{a}_k^{X_i}-{a}_k^{X_i}\right| &= \frac{2}{\pi^{1/4}}\left|\hat{\mathbb{E}}_{X_i}[\bar{h}_k(X_i)]-{\mathbb{E}}_{X_i}[\bar{h}_k(X_i)]\right|\\
        &\le \frac{\pi^{1/4}\epsilon}{8L_i(K_i+1)}.
        \label{eq:diff_ahat_atil}
    \end{align}
    We then have
    \begin{align}
        \left|\hat{f}_i(x)-{f}^{(K_i)}_i(x)\right|&=\left|\sum_{k=0}^{K_i} (\hat{a}_k^{X_i}-a_k^{X_i})h_k(x)\right| \nonumber \\
        &\le\sum_{k=0}^{K_i}\left| \hat{a}_k^{X_i}-a_k^{X_i}\right| \pi^{-1/4} \nonumber \\
        &\le\frac{\epsilon}{8L_i},
        \label{eq:diff_ftil_fhat}
    \end{align}
    where we use \eqref{eq:HFBnd} at the first inequality, and \eqref{eq:diff_ahat_atil} at the third inequality.
    Combining \eqref{eq:TriIneqf}, \eqref{eq:Diff_tilf_f}, and \eqref{eq:diff_ftil_fhat}, we get
    \begin{equation}
        \left|\hat{f}_i(x)-f_i(x)\right| \le \frac{\epsilon}{4L_i}.
    \end{equation}
    Then, this yields
    \begin{align}
        \left|F_i(x)-\hat{F}_i(x)\right| & \le \left|F_i(-L_i)\right| + \left|\int_{-L_i}^x f_i(y)dy - \int_{-L_i}^x \hat{f}_i(y)dy\right| \nonumber \\
        &\le \frac{\epsilon}{2} + \int_{-L_i}^x \left|\hat{f}_i(y)-f_i(y)\right| dy  \nonumber \\
        &\le \frac{\epsilon}{2} + \frac{\epsilon}{4L_i} \times (x+L_i) \nonumber \\
        & \le \epsilon
    \end{align}
    for $x\in(-L_i,L_i)$.

    To complete the proof, let us prove the statements on the success probability and complexity.
    Since the probability that each of $\proc{QMCI}\left(\mathcal{A}_{X_i},\bar{h}_k,\frac{\sqrt{\pi}\epsilon}{16L_i(K_i+1)},\frac{\delta}{K_i+1}\right)$ for $k=0,\cdots,K_i$ outputs $\hat{\mathbb{E}}_{X_i}[\bar{h}_k(X_i)]$ satisfying \eqref{eq:alphaErr} is at least $1-\frac{\delta}{K_i+1}$, the probability that all of them output such estimations is at least $\left(1-\frac{\delta}{K_i+1}\right)^{K_i+1} \ge 1-\delta$.
    The number of queries to $\mathcal{A}_{X_i}$ and $W_{\bar{h}_k}$ in $\proc{QMCI}\left(\mathcal{A}_{X_i},\bar{h}_k,\frac{\sqrt{\pi}\epsilon}{16L_i(K_i+1)},\frac{\delta}{K_i+1}\right)$ is
    \begin{equation}
        O\left(\frac{L_i(K_i+1)}{\epsilon}\log\left(\frac{K_i+1}{\delta}\right)\right)
    \end{equation}
    as implied by Theorem \ref{th:QMCI}. Summing this up for $k=0,\cdots,K_i$, we reach the total query number estimation in \eqref{eq:CompMarg}.
\end{proof}

\subsection{Estimating the risk measure}\label{subsec3-2}
Given the approximations of the CDFs $\hat{F}_1,\cdots,\hat{F}_d$ by the method above, we finally consider estimating the expected value $\mathbb{E}_{\bm{X}}\left[g\left(\bm{X}\right)\right]$ in \eqref{eq:TgtExp}.
We now assume that we have access to the controlled rotation oracle $W_\Phi$ for the function $\Phi:\mathbb{R}^d \rightarrow [0,1]$ expressed by
\begin{equation}
    \Phi(\bm{x})=\frac{1}{c_{\rm max}} g(\bm{x})c(\hat{F}_1(x_1),...,\hat{F}_d(x_d))
    \label{eq:PhiForm}
\end{equation}
with $c_{\rm max}:=\max_{\bm{u}\in[0,1]^d} c(\bm{u})$ and any functions $\hat{F}_1,\cdots,\hat{F}_d$ written in the form of \eqref{eq:hatF}. Then, we present Algorithm \ref{alg:RMEst} to estimate $\mathbb{E}_{\bm{X}}\left[g\left(\bm{X}\right)\right]$.

\begin{algorithm}[H]
\caption{$\proc{EstRM}(c,g,\epsilon,\delta)$: estimation of $\mathbb{E}_{F_{\bm{X}}}\left[g\left(\bm{X}\right)\right]$}\label{alg:RMEst}
\begin{algorithmic}[1]
\Require
\Statex Access to the state preparation oracles $\mathcal{A}_{X_1},\cdots,\mathcal{A}_{X_d}$.
\Statex Access to the rotation oracles $W_{\bar{h}_k}$ for $k=0,\cdots,K_i$ and $W_{\Phi}$ for $\Phi$ in \eqref{eq:PhiForm}.
\Statex $L_1,\cdots,L_d>0$ and $\epsilon,\delta\in(0,1)$.

\For{$i=1,\dots,d$}
\State Set $K_i$ as
\begin{equation}
    K_i=\left\lceil\left(\frac{16d c_\mathrm{max}^\prime L_i\gamma_i}{\epsilon}\right)^{4/(2r_i-1)}\right\rceil
    \label{eq:Ki2}
\end{equation}
with $c_\mathrm{max}^\prime$ given in Theorem \ref{th:RMEst}.

\State Run $\proc{EstMarg}\left(X_i,K_i,\frac{\epsilon}{2 d c_\mathrm{max}^\prime },\frac{\delta}{2d}\right)$ and let the output be $\hat{F}_i$.
\label{setp:EstMargInEstRM}

\EndFor

\State Run $\proc{QMCI}\left(\mathcal{A}^{\rm ind}_{\bm{X}},\hat{\Phi},\frac{\epsilon}{2c_\mathrm{max}},\frac{\delta}{2}\right)$, where
\begin{equation}
    \hat{\Phi}(\bm{x})=\frac{1}{c_{\rm max}} g(\bm{x})c(\hat{F}_1(x_1),...,\hat{F}_d(x_d)),
\end{equation}
and let the result be $\hat{\mathbb{E}}^\mathrm{ind}_{\bm{X}}(\hat{\Phi})$.
\label{step:LastQMCI}

\State Output $c_{\rm max}\hat{\mathbb{E}}^\mathrm{ind}_{\bm{X}}(\hat{\Phi})$.

\end{algorithmic}
\end{algorithm}

The following theorem is on the accuracy and complexity of this algorithm.\\

\begin{theorem}
    Let $\delta,\epsilon\in(0,1)$.
    Let $c:[0,1]^d\rightarrow\mathbb{R}$ be the density of copula and suppose that there exists $c^\prime_\mathrm{max}\in\mathbb{R}$ such that, for any $i=1,\cdots,d$ and any $\bm{u}\in[0,1]^d$,
    \begin{equation}
        \left|\frac{\partial}{\partial u_i}c(u_1,\cdots,u_d)\right|\le c^\prime_\mathrm{max}.
        \label{eq:cDerBound}
    \end{equation}
    Let $X_1,\cdots,X_d$ be real-valued random variables and suppose that all the assumptions in Theorem \ref{th:marg} are satisfied for every $X_i$ except that $\epsilon$ is replaced with $\frac{\epsilon}{2 d c_\mathrm{max}^\prime }$.
    Suppose that we have access to the rotation oracle $W_{\Phi}$ for any function in the form of \eqref{eq:PhiForm} with $g:\mathbb{R}^d\rightarrow[0,1]$.
    Then, with probability at least $1-\delta$, Algorithm \ref{alg:RMEst} outputs an $\epsilon$-approximation of $\mathbb{E}_{\bm{X}}\left[g\left(\bm{X}\right)\right]$, querying $W_{\bar{h}_k}$
    \begin{align}
        O\left(\frac{d^2c^\prime_\mathrm{max}LK^2}{\epsilon}\log\left(\frac{dK}{\delta}\right)\right)\label{eq:compRMEst1}
    \end{align}
    times, 
    $W_\Phi$
    \begin{equation}
        O\left(\frac{c_\mathrm{max}}{\epsilon}\log\left(\frac{1}{\delta}\right)\right)
        \label{eq:compRMEst2}
    \end{equation}
    times, and
    $\mathcal{A}_{X_i}$
    \begin{equation}
        O\left(\frac{d^2c^\prime_\mathrm{max}LK^2}{\epsilon}\log\left(\frac{dK}{\delta}\right)+\frac{c_\mathrm{max}}{\epsilon}\log\left(\frac{1}{\delta}\right)\right)
        \label{eq:compRMEst3}
    \end{equation}
    times, where $L:=\max_{i=1,\cdots,d} L_i$ and $K:=\max_{i=1,\cdots,d} K_i$.
    
    \label{th:RMEst}
\end{theorem}

\begin{proof}
    As implied by Theorem \ref{th:marg}, $\hat{F}_i$ obtained in step \ref{setp:EstMargInEstRM} in Algorithm \ref{alg:RMEst} satisfies
    \begin{equation}
        \forall x\in\mathbb{R}, \ \left|\hat{F}_i(x)-F_i(x)\right| \le \frac{\epsilon}{2 d c_\mathrm{max}^\prime }
        \label{eq:diff_Fhat_F_RM}
    \end{equation}
    with probability at least $1-\frac{\delta}{2d}$.
    If this holds, for any $\bm{x}\in\mathbb{R}^d$,
    \begin{align}
        & \left|\hat{\Phi}(\bm{x})-\Phi(\bm{x})\right| \nonumber \\
        &\le \frac{g(\bm{x})}{c_{\rm max}} \left|c(\hat{F}_1(x_1),...,\hat{F}_d(x_d))-c(F_1(x_1),...,F_d(x_d))\right| \nonumber \\
        &\le  \frac{c_\mathrm{max}^\prime g(\bm{x})}{c_\mathrm{max} }\sum_{i=1}^d \left|\hat{F}_i(x)-F_i(x)\right| \nonumber \\
        &\le  \frac{\epsilon}{2c_\mathrm{max}}
    \end{align}
    holds, where we use Taylor's theorem with \eqref{eq:cDerBound}.
    We then have
    \begin{align}
        &\left|{\mathbb{E}}^\mathrm{ind}_{\bm{X}}[\hat{\Phi}\left(\bm{X}\right)]-{\mathbb{E}}^\mathrm{ind}_{\bm{X}}\left[\Phi\left(\bm{X}\right)\right]\right| \nonumber \\
        &= \left|\frac{1}{\mathcal{N}_{\bm{X}}^2}\sum_{\bm{j}\in\mathcal{G}_{\bm{X}}} \left(\hat{\Phi}(\bm{x}^{(\bm{j})})-\Phi(\bm{x}^{(\bm{j})})\right)f^\mathrm{ind}(\bm{x}^{(\bm{j})})\right| \nonumber \\
        &\le \frac{1}{\mathcal{N}_{\bm{X}}^2}\sum_{\bm{j}\in\mathcal{G}_{\bm{X}}} \left|\hat{\Phi}(\bm{x}^{(\bm{j})})-\Phi(\bm{x}^{(\bm{j})})\right|f^\mathrm{ind}(\bm{x}^{(\bm{j})}) \nonumber \\
        &\le \frac{\epsilon}{2c_\mathrm{max}}.
        \label{eq:diff_EtilPhihat_EtilPhi}
    \end{align}
    On the other hand, $\proc{QMCI}\left(\mathcal{A}^\mathrm{ind}_{\bm{X}},\hat{\Phi},\frac{\epsilon}{2c_\mathrm{max}},\frac{\delta}{2}\right)$ outputs $\hat{\mathbb{E}}_{\bm{X}}[\hat{\Phi}(\bm{X})]$ such that
    \begin{equation}
        \left|\hat{\mathbb{E}}^\mathrm{ind}_{\bm{X}}[\hat{\Phi}(\bm{X})]-{\mathbb{E}}^\mathrm{ind}_{\bm{X}}[\hat{\Phi}\left(\bm{X}\right)]\right|\le \frac{\epsilon}{2c_\mathrm{max}}
        \label{eq:diff_Ehat_EtilPhihat}
    \end{equation}
    with probability at least $1-\delta$.
    Combining \eqref{eq:diff_EtilPhihat_EtilPhi} and \eqref{eq:diff_Ehat_EtilPhihat}, we get
    \begin{align}
        & \left|c_\mathrm{max}\hat{\mathbb{E}}^\mathrm{ind}_{\bm{X}}[\hat{\Phi}(\bm{X})]-\mathbb{E}_{\bm{X}}\left[g\left(\bm{X}\right)\right]\right| \nonumber \\
        &\le  \ c_\mathrm{max}\left(\left|\hat{\mathbb{E}}^\mathrm{ind}_{\bm{X}}[\hat{\Phi}(\bm{X})]-{\mathbb{E}}^\mathrm{ind}_{\bm{X}}[\hat{\Phi}\left(\bm{X}\right)]\right| \right. \nonumber \\
        &\qquad\qquad\left.+\left|{\mathbb{E}}^\mathrm{ind}_{\bm{X}}[\hat{\Phi}\left(\bm{X}\right)]-{\mathbb{E}}^\mathrm{ind}_{\bm{X}}[\Phi\left(\bm{X}\right)]\right|\right) \nonumber \\
        &\le  \ \epsilon.
    \end{align}
    This holds if every $\hat{F}_i$ satisfies \eqref{eq:diff_Fhat_F_RM} and \eqref{eq:diff_Ehat_EtilPhihat} holds, whose probability is at least
    \begin{equation}
        \left(1-\frac{\delta}{2d}\right)^d\left(1-\frac{\delta}{2}\right)\ge 1-\delta.
    \end{equation}
    
    Lastly, let us evaluate the query complexity of the algorithm.
    In estimating each $\hat{F}_i$, $\mathcal{A}_{X_i}$ and $\{W_{h_k}\}_k$ are queried
    \begin{equation}
        O\left(\frac{dc^\prime_\mathrm{max}L_iK_i^2}{\epsilon}\log\left(\frac{dK_i}{\delta}\right)\right)
    \end{equation}
    times, and, in estimating all of $\hat{F}_1,\cdots,\hat{F}_d$, this is multiplied by $d$.
    In $\proc{QMCI}\left(\mathcal{A}^{\rm ind}_{\bm{X}},\hat{\Phi},\frac{\epsilon}{2c_\mathrm{max}},\frac{\delta}{2}\right)$, $\mathcal{A}^{\rm ind}_{\bm{X}}$ and $W_{\hat{\Phi}}$ are called the number of times of order \eqref{eq:compRMEst2}.
    In $\mathcal{A}_{\bm{X}}^{\rm ind}$, $\mathcal{A}_{X_1},\cdots,\mathcal{A}_{X_d}$ are called once each and $d$ times in total.
    Combining these observations, we reach the query number bounds in \eqref{eq:compRMEst1}, \eqref{eq:compRMEst2}, and \eqref{eq:compRMEst3}.
    
\end{proof}

Substituting $K=\left\lceil(16d c_\mathrm{max}^\prime L\gamma/\epsilon)^{4/(2r-1)}\right\rceil$ for $\gamma:=\max_{i=1,\cdots,d} \gamma_i$ and $r:=\min_{i=1,\cdots,d} r_i$ into \eqref{eq:compRMEst3}, we obtain the number of queries to $\mathcal{A}_{X_i}$ scaling on accuracy $\epsilon$ as
\begin{align}
    \tilde{O}\left((1/\epsilon)^{1+\frac{8}{2r-1}}\right).
\end{align}
If the smoothness parameter $r$ is high, the query complexity asymptotically approaches $\tilde{O}(1/\epsilon)$, and the algorithm provides nearly quadratic quantum speed-up.

Finally, we discuss how to utilize the algorithm $\proc{EstRM}$ to calculate risk measures.
To calculate ${\rm VaR}_S(\alpha)$, we first construct a procedure to compute $\mathrm{Pr}(S \ge l)$ for $S=X_1+\cdots+X_d$ and $l\in\mathbb{R}$ based on Algorithm \ref{alg:RMEst}, with $g$ set to
\begin{align}
    g_{\rm VaR}(\bm{x};l)=\mathbbm{1}_{s \ge l}.
\end{align}
Then, for $\alpha\in(0,1)$, we search the smallest $l_\alpha$ satisfying
\begin{align}
    \mathrm{Pr}(S \ge l_\alpha)\le 1-\alpha,
\end{align}
which corresponds to VaR$_S(\alpha)$, using a root-finding method.
For example, we can use binary search, which increases the total query complexity only by a logarithmic factor.
For TVaR, using the calculated VaR $l_\alpha$, we run Algorithm \ref{alg:RMEst} setting $g$ to
\begin{align}
    g_{\rm TVaR}(\bm{x};l_\alpha)=\sum_{i=1}^d \frac{x_i}{x_{\rm max}} \mathbbm{1}_{s \ge l_\alpha},
    \label{eq:gTVaR}
\end{align}
where $x_{\rm max}:=\max_{i=1,...,d}x_i^{(N_i)}$, and let the output multiplied by $x_{\rm max}$ be the estimation of TVaR.

\section{Conclusion}\label{sec:Concl}
We developed a quantum algorithm to conduct risk aggregation, a fundamental task in financial institutions.
In particular, we considered a standard approach that uses copula, a tool to generate a joint distribution from marginal distributions, to express the structure of the correlations among the risk variables.
However, the common classical algorithm for this approach is based on sampling and rearranging, which cannot be straightforwardly applied in the quantum setting, and thus no efficient quantum algorithm was known.
To address this challenge, we %have
developed a quantum method consisting of two steps.
First, we estimate each marginal CDF based on OSDE.
As each coefficient in the orthogonal series is expressed as the expected value of the orthogonal function, we apply QMCI to estimate it.
Then, we plug the estimated CDFs into the copula to get the joint distribution, from which we calculate the risk measure using QMCI once more.
Using this method, we achieved a nearly quadratic quantum speed-up of risk aggregation, compared to the classical method for smooth marginal PDFs. 

\section*{Acknowledgment}

This work is supported by MEXT Quantum Leap Flagship Program (MEXTQLEAP) Grant
No. JPMXS0120319794, and JST COI-NEXT program Grant No. JPMJPF2014. KM is also supported by JSPS KAKENHI Grant no. JP22K11924.

\begin{appendices}
\section{Copulas}\label{sec:Copula}
For readers unfamiliar with copula, we introduce some widely used types: 

\begin{itemize}
    \item Gaussian copula: for any $n \times n$ correlation matrix $R=(\rho_{ij})$, the Gaussian copula is defined by
    \begin{equation}
        C=\Phi_{n,R}(\Phi_\mathrm{SN}^{-1}(\cdot),\cdots,\Phi_\mathrm{SN}^{-1}(\cdot)),
    \end{equation}
    where $\Phi_\mathrm{SN}$ is the CDF for the standard normal distribution and $\Phi_{n,R}$ is the CDF for the $n$-variate normal distribution with zero mean vector and covariance matrix equal to $R$. This is one of the most widely used copulas in practice because of its simplicity, but its tail dependence can be too weak in some situations.

    \item T-copula: with a parameter $\nu>0$ called the number of degrees of freedom (DOF) and a positive-define matrix $\Sigma\in\mathbb{R}^{n \times n}$, a T-copula is defined by
    \begin{equation}
        C=t_{n,\nu,\Sigma}(t_\nu^{-1}(\cdot),\cdots,t_\nu^{-1}(\cdot)),
    \end{equation}
    where $t_\nu$ is the CDF for the Student's $t$ distribution with the number of DOF $\nu$, and $t_{n,\nu,\Sigma}$ is the CDF for the $n$-variate $t$ distribution with the number of DOF $\nu$ and the covariance matrix $\Sigma$. This copula exhibits stronger tail dependence compared to Gaussian copulas and is preferred in certain scenarios.
\end{itemize}

\section{Evaluation of \texorpdfstring{$\mathcal{H}_{k,L}$}{Lg}}\label{sec:HFIntEval}

To evaluate $\mathcal{H}_{k,L}(x)$, we first consider how to evaluate
\begin{equation}
    \tilde{\mathcal{H}}_{l}(x):=\frac{1}{\sqrt{2\pi}}\int_{-\infty}^x y^l e^{-y^2/2} dy
\end{equation}
for $l\in\mathbb{N}\cup\{0\}$.
We can utilize a recurrence formula
\begin{equation}
    \tilde{\mathcal{H}}_{l}(x)=-\frac{1}{\sqrt{2\pi}}x^{l-1}e^{-x^2/2}+(l-1)\tilde{\mathcal{H}}_{l-2}(x), l\ge2
\end{equation}
with
\begin{equation}
    \tilde{\mathcal{H}}_{1}(x)=-\frac{1}{\sqrt{2\pi}}e^{-x^2/2}, \tilde{\mathcal{H}}_{0}(x)=\Phi_{\rm SN}(x).
\end{equation}
Thus, we can compute $\tilde{\mathcal{H}}_{l}(x)$ by evaluating the terms in the form of $x^le^{-x^2/2}$ and $\Phi_{\rm SN}(x)$.
We can evaluate the former on a quantum computer using circuits for elementary functions \cite{haner2018optimizing}.
On the latter, there are some approximation formulas with elementary functions (see, e.g., \cite{abramowitz1968handbook}), and thus it is also evaluated on a quantum computer.

\end{appendices}

\bibliography{cite}
\bibliographystyle{junsrt}

\end{document}